\documentclass[11pt,a4paper]{article}

\newcommand{\Preprint}[2]{#1}
\newcommand{\Paper}{\Preprint{paper }{technical note }}

\usepackage{graphicx}
\usepackage{amsmath}
\usepackage{amsthm}
\usepackage{todonotes}
\usepackage{color}

\newtheorem{definition}{Definition}
\newtheorem{property}{Property}
\newtheorem{theorem}{Theorem}
\newtheorem{lemma}{Lemma}

\newcommand{\Closure}[1]{\mathit{cl}(#1)}
\newcommand{\Frontier}[1]{{\partial\Closure{#1}}}
\newcommand{\Interior}[1]{\mathit{int}(#1)}
\newcommand{\Lie}[2]{\mathcal{L}_{#1} #2}
\newcommand{\Flow}{\varphi}
\newcommand{\Sign}[1]{\mathit{sgn}(#1)}

\newcommand{\Changed}[1]{#1}

\usepackage{amssymb}

\title{Converse Theorems\\for Safety and Barrier Certificates}
\author{Stefan Ratschan\thanks{Institute of Computer Science, Czech Academy of Sciences, ORCID: 0000-0003-1710-1513.}\thanks{The research published in this paper was supported by GA{\v C}R grant GA15-14484S and by the long-term strategic development financing of the Institute of Computer Science (RVO:67985807).
We thank Peter Franek for indispensable discussions throughout the work on this paper.}}

\begin{document}
\maketitle

\begin{abstract}
  An important tool for proving safety of dynamical systems is the notion of a barrier certificate. In this \Paper we prove that every robustly safe ordinary differential equation has a barrier certificate. Moreover, we show a construction of such a barrier certificate based on a set of states that is reachable in finite time. 

\end{abstract}

\section{Introduction}

An important property of a control system is safety, that is, the property that the system, when starting from a certain set of initial states, always stays in a set of safe states. The notion of a barrier certificate~\cite{Prajna:04,Taly:09} proves safety by separating the set of initial states from the set of states that are not safe using a transversality condition. This immediately opens the question whether every safe system has a barrier certificate. In this \Paper we give an affirmative answer for robustly safe ordinary differential equations. The construction is based on an inductive certificate for safety. Every safe system has such a certificate, but it is formed by a set of states reachable in unbounded time. 
As a second result of this paper, we show that a certain set of states reachable in finite time always  forms such an inductive certificate for safety, as well.

Converse theorems for barrier certificates have received attention since their introduction. Prajna and Rantzer~\cite{Prajna:05,Prajna:07} proved a converse theorem that holds under some Slater-like condition. Wisniewski and Sloth introduced the idea of exploiting robustness for such a converse theorem~\cite{Wisniewski:16}. Their result does not need the Slater-like condition, but is restricted to Morse-Smale vector fields. 

The result in this paper is based on robustness arguments, too. However, it neither needs a Slater-like condition, nor is restricted to Morse-Smale vector fields. Moreover, unlike the earlier results, which are based on measure theory and Morse-Smale theory, respectively, the proofs in this paper are based on basic mathematical tools, only.


\section{Problem Definition}

\begin{definition}
A \emph{safety verification problem} is a triple $(f, I, U)$ where
\begin{itemize}
\item $f: \mathbb{R}^n\rightarrow \mathbb{R}^n$, and $f$ is smooth,
\item $I\subseteq \mathbb{R}^n$ (the set of \emph{initial states}),  and
\item $U\subseteq \mathbb{R}^n$ (the set of \emph{unsafe states}).
\end{itemize}
\end{definition}

The vector field $f$ defines an ordinary differential equation whose solution we will describe using the flow of $f$:

\begin{definition}
\label{def:1} For a smooth $f: \mathbb{R}^n\rightarrow \mathbb{R}^n$, and $t\in \mathbb{R}$, we denote the state that system $\dot{x}=f(x)$ reaches after time $t$ from $x$ by   $\Flow_f(x, t)$. 
\end{definition}

Note that we also allow $t$ to be negative, and we have $\Flow(x, t)=y$ iff $\Flow(y, -t)=x$.

Based on the flow $\Flow$ we introduce the following notation for sets of reachable states:

\begin{definition}
\label{def:3}
For a smooth $f: \mathbb{R}^n\rightarrow \mathbb{R}^n$, for sets $X\subseteq \mathbb{R}^n$ and $T\subseteq \mathbb{R}$, $R_f^{T}(X):= \{ \Flow_f(x,t) \mid t\in T, x\in X\}$. For $t\in \mathbb{R}$, 
 $R_f^t(X):= R_f^{\{t\}}(X)$, and $R_f(X):= R_f^{\mathbb{R}^{\geq 0}}(X)$.
\end{definition}

Here, if clear from the context, we will drop the subscript $f$.

\begin{definition}
\label{def:4}
  A safety verification problem $(f, I, U)$ is \emph{safe} iff $R_f(I)\cap U=\emptyset$.
\end{definition}


In this paper, we are interested in objects that can serve as evidence that a given safety verification problem is safe. One possibility of such objects is the following:

\begin{definition}
\label{def:2}
A set $V\subseteq\mathbb{R}^n$ is a \emph{safety certificate} of a safety verification problem $(f, I, U)$ iff
 \begin{itemize}
 \item $I\subseteq V$
 \item $R_f(V)\subseteq V$
 \item $V\cap U=\emptyset$
 \end{itemize}
\end{definition} 

A set that fulfills the first two conditions is also called \emph{inductive invariant} by the verification community, and a \emph{positively invariant set} by the dynamical systems literature~\cite{Khalil:02,Hirsch:03}. 

Safety certificates serve as a proof of safety:
\begin{property}
\label{prop:correct}
If a safety verification problem $(f, I, U)$ has a safety certificate, then it is safe.  
\end{property}






Safety certificates are also complete, that is, there is a converse of Property~\ref{prop:correct}:

\begin{property}
\label{prop:converse}
  If a safety verification problem $(f, I, U)$ is safe, then it has a  safety certificate, $R_f^{\mathbb{R}^\geq 0}(I)$.
\end{property}


The definitions presented above are not robust against small changes of the safety verification problem. For example, a safety verification problem might have a safety certificate, but tiny changes of the safety verification problem might result in it not being safe. In order to avoid this, we use robust analogies of the definitions above.
Here we measure deviation from nominal behavior by the Euclidean norm, denoted by $|| \cdot ||$.

\begin{definition}
A function $x:\mathbb{R}^{\geq 0}\rightarrow\mathbb{R}^n$ is \emph{an $\varepsilon$-solution} of an ordinary differential equations $\dot{x}=f(x)$ iff for all $t\geq 0$,  $||f(x(t))-\dot{x}(t)|| \leq\varepsilon$. 
\end{definition}

We use the following robust versions of the reach sets introduced in Definition~\ref{def:3}:

\begin{definition}
For a smooth $f: \mathbb{R}^n\rightarrow \mathbb{R}^n$, for sets $X\subseteq \mathbb{R}^n$ and $T\subseteq \mathbb{R}$, 
$R_{f,\varepsilon}^{T}(X):= \{ x(t) \mid t\in T, x \text{ is an } \varepsilon\text{-solution of } f \text{ with } x(0)\in X \}$. For $t\in\mathbb{R}$, $R_{f,\varepsilon}^t(X):= R_{f,\varepsilon}^{\{t\}}(X)$, and $R_{f,\varepsilon}(X)= R_{f,\varepsilon}^{\mathbb{R}^{\geq 0}}$. 
\end{definition}

Again, if clear from the context, we will drop the subscript $f$.

Also Definition~\ref{def:4} has a robust version:\nopagebreak[3]%
\begin{definition}
A safety verification problem $(f, I, U)$ is \emph{robustly safe} iff there is an $\varepsilon>0$ such that $R_{f, \varepsilon}(I)\cap U=\emptyset$. We will call an $\varepsilon>0$ fulfilling this property a \emph{robustness margin} of $(f, I, U)$.
\end{definition}



We also provide a robust version of Definition~\ref{def:2}:
\begin{definition}
\label{def:5}
A set  $V\subseteq\mathbb{R}^n$ is an $\varepsilon$-\emph{robust  safety certificate} of a safety verification problem $(f, I, U)$ iff
 \begin{itemize}
 \item $I\subseteq V$
 \item $R_{f,\varepsilon}(V)\subseteq V$
 \item $V\cap U=\emptyset$
 \end{itemize}
\end{definition} 

We call a  safety certificate $V$ robust iff there is an $\varepsilon>0$ such that $V$ is a $\varepsilon$-robust safety certificate.

There is a robust version of Property~\ref{prop:correct}:
\begin{property}
\label{prop:correct_robust}
If a safety verification problem $(f, I, U)$ has a robust safety certificate, then it is robustly safe.  
\end{property}

Also completeness holds in analogy to Property~\ref{prop:converse}, that is, there is a converse of Property~\ref{prop:correct_robust}:

\begin{property}
\label{prop:converse_robust}
  If a safety verification problem $(f, I, U)$ is robustly safe, then it has a  robust safety certificate, $R^{\mathbb{R}^{\geq 0}}_{f,\varepsilon}(I)$, where $\varepsilon$ is a robustness margin of $(f, I, U)$.
\end{property}

Safety certificates have the disadvantage that checking them still needs the computation of a reachable set. This can be avoided by the following definition.

\begin{definition}[Prajna and Jadbabaie~\cite{Prajna:04}]
\label{def:6}
 A differentiable function  $\beta:\mathbb{R}^n\rightarrow\mathbb{R}$ is a \emph{barrier certificate} of a safety verification problem $(f, I, U)$ iff
\begin{itemize}
\item $\forall x \;.\; x\in I\Rightarrow \beta(x)\geq 0$
\item $\forall x \;.\; \beta(x)=0\Rightarrow L_f(\beta)(x)>0$, that is $(\nabla \beta(x))^T f(x) >0$
\item $\forall x \;.\; \beta(x)\geq 0\Rightarrow x\not\in U$
\end{itemize}
\end{definition}

Also barrier certificates serve as a proof of safety:
\begin{property}
\label{prop:correct_barrier}
If a safety verification problem $(f, I, U)$ has a barrier certificate, then it is safe.  
\end{property}

Converse theorems for this property are an active area of research~\cite{Prajna:05,Prajna:07,Wisniewski:16}, and at the same time the main topic of this paper. 

Summarizing, Definitions~\ref{def:2},~\ref{def:5}, and~\ref{def:6} provide three different forms of certificates for safety verification problems, as ensured by corresponding Properties~\ref{prop:correct},~\ref{prop:correct_robust}, and~\ref{prop:correct_barrier}, respectively. 
The contributions of this paper are converse theorems for the latter two cases. The state of the art is: 
\begin{itemize}
\item For Definition~\ref{def:5}, there is a very simple converse theorem of the corresponding Property~\ref{prop:correct_robust}. \Changed{This converse theorem is Property~\ref{prop:converse_robust} which has the disadvantage of being based on an infinite-time reach set.}
\item For Definition~\ref{def:6}, it is known that the literal converse of corresponding Property~\ref{prop:correct_barrier} does not hold, that is, that is, there are safety verification problems that are safe, but do not have a barrier certificate~\cite[Example 3]{Taly:09}. Still, there are partial converse theorems~\cite{Prajna:07,Wisniewski:16}. The most general result~\cite{Wisniewski:16} proves a converse for Morse-Smale vector fields under the assumption of robustness.
\end{itemize}

The main result of this paper is a converse theorem of Property~\ref{prop:correct_barrier} that still assumes robustness, but is not restricted to Morse-Smale vector fields:

\begin{theorem}
\label{thm:3}
  Assume a safety verification problem $(f, I, U)$ that is robustly safe, such that the closures of $I$ and $U$ are  disjoint, and the set of safe states $\mathbb{R}^n\setminus U$ is bounded. Then the safety verification problem has a  barrier certificate.
\end{theorem}

The second result is the following converse of Property~\ref{prop:converse_robust} that is based on a finite time reach set:

\begin{theorem}
\label{thm:1}
  If a safety verification problem $(f, I, U)$ is robustly safe with robustness margin~$\varepsilon$ and the set of safe states $\mathbb{R}^n\setminus U$ is bounded, then for all $\Delta>0$ there is a $t\geq 0$ such that  $R_{f,\frac{\varepsilon}{2}}^{[0,\Delta]}(R_{f,\varepsilon}^{[0,t]}(I))$ is an $\frac{\varepsilon}{2}$-robust safety certificate. 
\end{theorem}


In the following section we first summarize some facts that directly follow from the literature. Then, in Section~\ref{sec:converse} we prove Theorem~\ref{thm:3}, and in Section~\ref{sec:finite-time}, Theorem~\ref{thm:1}. Section~\ref{sec:conclusion} concludes the paper.

\section{Preliminaries}
\label{sec:preliminaries}

In this section we summarize some facts that directly follow from the literature. We first relate the notion of an $\varepsilon$-solution to controllability:

\begin{property}
\label{prop:1}
Given a smooth $f:\mathbb{R}^n\rightarrow\mathbb{R}^n$ and $\Delta>0$, if there is a control input $u: [0, \Delta]\rightarrow [-\varepsilon,\varepsilon]^n$ s.t. $\dot{x}=f(x)+u(t)$ has a solution with $x(0)=p$ and $x(\Delta)=p'$, then $p'\in R^{\Delta}_{f, \varepsilon}(\{p\})$.
\end{property}

Hence we can apply results from controllability to our context. The following property directly follows from this:

\begin{property}
\label{prop:control}
Given an ODE $\dot{x}=f(x)$ with smooth $f:\mathbb{R}^n\rightarrow\mathbb{R}^n$, $\varepsilon$, $\hat{\varepsilon}$ with $0\leq\varepsilon<\hat{\varepsilon}$, $\Delta>0$,
$p$ and $p'$ such that $p'\in R^{\Delta}_{f, \varepsilon}(\{p\})$, there is a  $\delta>0$ such that  every point in the $\delta$-neighborhood of $p'$ is in $R_{f, \hat{\varepsilon}}^\Delta(\{p\})$.
\end{property}

The usual proof of (a generalization of) this property~\cite[Proposition 3.3]{Nijmeijer:90},~\cite[Proposition 11.2]{Sastry:99} uses the inverse function theorem to map each point $p^*$ in the neighborhood of $p'$ to a vector 
$\xi_1,\dots,\xi_n$ that generates an input of the form needed by Property~\ref{prop:1} to steer the system into  $p^*$. Due to a version of the inverse function theorem~\cite[Theorem 2.9.4]{Hubbard:01} that bounds the size of this neighborhood from below based on a Lipschitz constant for the derivative of the given function, the size $\delta$ of the neighborhood in Property~\ref{prop:control} can be bounded from below over all elements $p$ and $p'$ of a compact set:

\begin{lemma}
\label{lem:1}
Given an ODE $\dot{x}=f(x)$ with smooth $f:\mathbb{R}^n\rightarrow\mathbb{R}^n$, and a compact set $\Omega\subseteq\mathbb{R}^n$. Then, for every $\varepsilon>0$, $\Delta>0$ there is a $\delta>0$ such that for all $p, p'\in\Omega$ such that
 $p'\in \Closure{R^\Delta_{f,\frac{\varepsilon}{2}}(\{p\})}$, every point in the $\delta$-neighborhood of $p'$ is in $R^\Delta_{f,\varepsilon}(\{p\})$.
\end{lemma}

\section{Robust Converse Theorem for Barrier Certificates}
\label{sec:converse}

In this section, we will prove Theorem~\ref{thm:3}. Throughout the section, we will assume a safety verification problem $(f, I, U)$ that fulfills the premises of this theorem: \Changed{robust safety, disjointness of the closures of $I$ and $U$, and boundedness of the set of safe states $\mathbb{R}^n\setminus U$}. We will denote by $\mu$ the robustness margin of $(f, I, U)$, and by $V$ a $\mu$-robust safety certificate whose existence is ensured by Property~\ref{prop:converse_robust}. We will construct a barrier certificate from the boundary $\Frontier{V}$ of the closure of $V$. We will first prove that the solutions of $f$ starting in $\Frontier{V}$ do not have any further intersection with $\Frontier{V}$---even in reverse time:


\begin{lemma}
\label{lem:14}
  For every $x\in R_f^{\mathbb{R}}(\Frontier{V})$ there is precisely one $t\in\mathbb{R}$ such that $\Flow_f(x, t)\in\Frontier{V}$
\end{lemma}

\begin{proof}
By the definition of $R_f^{\mathbb{R}}(\Frontier{V})$ there is at least one such $t$. It suffices to prove that there is not more than one. For this we assume that there are $t_1$ and $t_2$ such that $t_1\neq t_2$, and both $\Flow_f(x, t_1)$ and $\Flow_f(x, t_2)$ are in $\Frontier{V}$. W.l.o.g. assume that $t_1<t_2$. Let $\hat{t}$ be such that $t_1<\hat{t}<t_2$. 

Then $\Flow(x, t_2)= \Flow(\Flow(x, \hat{t}), t_2-\hat{t})$ and so $\Flow(x, t_2)\in R^{t_2-\hat{t}}_{f, 0}(\{\Flow(x, \hat{t})\})$. Therefore, by Property~\ref{prop:control}, there is a neighborhood of $\Flow(x, t_2)$ such that every point in this neighborhood is $\mu$-reachable from $\Flow(x, \hat{t})$. The same argument applied in reverse time to $t_1$ and $\hat{t}$ ensures a neighborhood of $\Flow(x, t_1)$ such that $\Flow(x, \hat{t})$ is $\mu$-reachable from every point in this neighborhood. Let $x_1$ be in the neighborhood of $\Flow(x, t_1)$ such that $x_1\in V$, and let $x_2$ be in the neighborhood of $\Flow(x, t_2)$ such that $x_2\not\in V$. Then there is a $\mu$-solution from $x_1\in V$ over $\Flow_f(x, \hat{t})$ to $x_2\not\in V$. This is in contradiction to the fact that $V$ is a $\mu$-robust safety certificate. 
\end{proof}

\Changed{We will now introduce a function $\nu(x)$ from which we will later obtain the barrier certificate by smooth approximation.}
Let $\nu(x)$ be such that $\nu(x)$ is defined on $R_f^{\mathbb{R}}(\Frontier{V})$ and such that $\nu(x):= -t$ with $t$ being the unique $t$ with $\Flow_f(x, t)\in\Frontier{V}$ ensured by Lemma~\ref{lem:14}. Note that $\nu(x)$ is positive for all points of the  interior of $V$ for which it is defined and negative for all points of the complement of the closure of $V$ for which it is defined. Moreover, $\Frontier{V}$ is in the interior of the domain of definition of $\nu$.
%
%
%
Due to the definition of $\nu$ we also have:
\Changed{
\begin{property}
\label{prop:3}
    For $x\in R_f^{\mathbb{R}}(\Frontier{V})$, $\nu(\varphi(x, t))=\nu(x)+t$.
\end{property}}

In general, functions defined based on the length of solutions leading to some set are not continuous. But using the robustness of $V$ we will be able to prove continuity:

\begin{lemma}
\label{lem:16}
The function $\nu$ is continuous in an open neighborhood of $\Frontier{V}$.
\end{lemma}

Before formally proving this lemma, we explain the idea of the proof. We assume that $\nu$ is not continuous let us say at a point $x$, and derive a contradiction. Let $y$ be the point where the solution from $x$ enters $V$. Then, due to non-continuity, the situation shown in Figure~\ref{fig:intuition} arises: There is a point $x'$ close to $x$ such that the solution from $x'$ enters a close neighborhood of $y$ (point $\hat{y}$), but then enters $V$ at a different point further away. Letting $y'$ be close to that entrance point but still outside of $V$ we can perturb the trajectory from $\hat{y}$ to $y'$ by moving $\hat{y}$ to $y$. The resulting trajectory leads from $y$ to $y'$ and hence leaves $V$ which is a contradiction to $V$ being a robust safety certificate.

\begin{figure}[ht]
  \centering
  \Preprint{\includegraphics[width=9cm,clip,trim={8cm 12cm 8cm 1cm}]{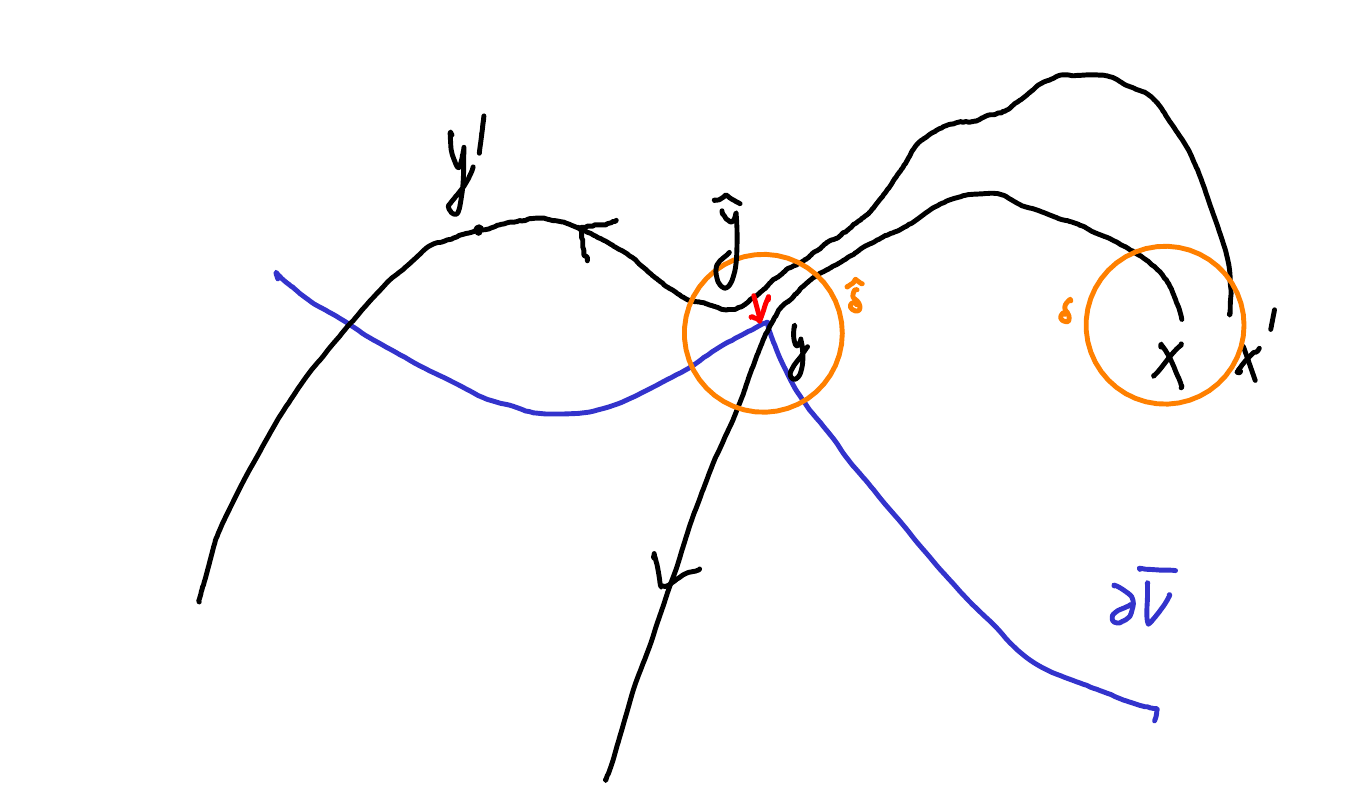}}{\includegraphics[width=9cm]{continuity_nu_new.pdf}}

  \caption{Intuition of the Proof of Lemma~\ref{lem:16}}
  \label{fig:intuition}
\end{figure}

The figure illustrates the situation where $x\not\in V$, but we do not assume this in the proof by exploiting the fact that the flow $\Flow$ also allows negative time.

\begin{proof}
The function $\nu$ being continuous in an open neighborhood of $\Frontier{V}$ means
that for all elements $x$ of this neighborhood, for all $\varepsilon>0$ there exists $\delta>0$ s.t. for all $x'$ with $||x'-x||<\delta$, $|\nu(x')-\nu(x)|<\varepsilon$. We assume non-continuity of $\nu$ in every open neighborhood of $\Frontier{V}$, and derive a contradiction.

Since the set of states $\mathbb{R}^n\setminus U$ is bounded and the safety verification problem $(f, I, U)$ is robustly safe, the safety certificate $V$ is bounded, too. Let $N$ be an open, but bounded, neighborhood of $\Frontier{V}$.  Due to our assumption, 
 $\nu$ is non-continuous in $N$, that is, there is an $x\in N$, $\varepsilon>0$ s.t. $\forall \delta>0 \exists x'$ s.t. $||x-x'||\leq \delta$, $|\nu(x)-\nu(x')|\geq\varepsilon$. Choose such an $x$ and $\varepsilon$.

\Changed{Since $N$ is open, but bounded, $\Closure{N}$ is compact}. Now let $\hat{\delta}^+$ be the constant ensured by Lemma~\ref{lem:1} after choosing $\Omega$ as $\Closure{N}$, $\varepsilon$ as $\frac{\mu}{2}$, and  $\Delta$ as the constant $\varepsilon$ chosen above. Let $\hat{\delta}^-$ be the corresponding constant for reverse dynamics $-f$. Choose $\hat{\delta}=\min\{ \hat{\delta}^+, \hat{\delta}^-\}$. Then for all $p, p'\in N$ such that
$p'\in \Closure{R^{\varepsilon}_{f,\frac{\mu}{4}}(\{p\})}$, every point in the $\hat{\delta}$-neighborhood of $p'$ is in $R^{\varepsilon}_{f,\frac{\mu}{2}}(\{p\})$. Also for all $p, p'\in N$ such that
 $p'\in \Closure{R^{\varepsilon}_{-f,\frac{\mu}{4}}(\{p\})}$, every point in the $\hat{\delta}$-neighborhood of $p'$ is in $R^{\varepsilon}_{-f,\frac{\mu}{2}}(\{p\})$. We note these properties and refer to them below by $(*)^+$, and $(*)^-$, respectively.

Since the flow $\Flow$ is continuous we can choose a $\delta$ s.t. for all $x'$ with $||x-x'||<\delta$,  $||\Flow(x, -\nu(x))- \Flow(x', -\nu(x))||<\hat{\delta}$. Due to our assumption of non-continuity of $\nu$, for this $\delta$ there is an $x'$ s.t. $||x-x'||\leq \delta$, $|\nu(x)-\nu(x')|\geq\varepsilon$. 

Let $y:=\Flow_f(x,-\nu(x))$, $\hat{y}:= \Flow_f(x',-\nu(x))$, $y':= \Flow_f(\hat{y}, \Sign{\nu(x)-\nu(x')}\varepsilon)$. Then $||y-\hat{y}||<\hat{\delta}$, and $y'=\Flow_f(\Flow_f(x',-\nu(x)), \Sign{\nu(x)-\nu(x')}\varepsilon)=\Flow_f(x',-\nu(x)+\Sign{\nu(x)-\nu(x')}\varepsilon)$.
Hence, by Property~\ref{prop:3}, $\nu(y')= \nu(\Flow_f(x',-\nu(x)+\Sign{\nu(x)-\nu(x')}\varepsilon))=\nu(x')-\nu(x)+\Sign{\nu(x)-\nu(x')}\varepsilon$.

We have two cases:
\begin{itemize}
\item If $\nu(x)-\nu(x')\geq 0$, then $\nu(x)-\nu(x')=|\nu(x)-\nu(x')|\geq \varepsilon$, and 
  \begin{align*}
\nu(y')&=&\\
\nu(x')-\nu(x)+\Sign{\nu(x)-\nu(x')}\varepsilon&=&\\
-(\nu(x)-\nu(x'))+\varepsilon &\leq& 0.
  \end{align*}
  Since $y'= \Flow_f(\hat{y}, \Sign{\nu(x)-\nu(x')}\varepsilon)=\Flow_f(\hat{y}, \varepsilon)$, $\hat{y}=\Flow_{-f}(y',\varepsilon)$, and hence $\hat{y}\in \Closure{R^{\varepsilon}_{-f,\frac{\mu}{4}}(\{y'\})}$. So by Property~$(*)^-$ noted above, $y\in R^{\varepsilon}_{-f,\frac{\mu}{2}}(\{y'\})$, and so $y'\in R^{\varepsilon}_{f,\frac{\mu}{2}}(\{y\})$. But $\nu(y)=0$ and $\nu(y')\leq 0$. Hence $y\in\Closure{V}$ and
  $y'\not\in\Interior{V}$.

\item If $\nu(x)-\nu(x')\leq 0$, then $-(\nu(x)-\nu(x'))=|\nu(x)-\nu(x')|\geq \varepsilon$, and 
  \begin{align*}
\nu(y')&=&\\
\nu(x')-\nu(x)+\Sign{\nu(x)-\nu(x')}\varepsilon&=&\\
-(\nu(x)-\nu(x'))-\varepsilon &\geq& 0.
  \end{align*}
  

  Since $y'= \Flow_f(\hat{y}, \Sign{\nu(x)-\nu(x')}\varepsilon)=\Flow_f(\hat{y}, -\varepsilon)$, $\hat{y}= \Flow_{f}(y', \varepsilon)$, and hence $\hat{y}\in \Closure{R^{\varepsilon}_{f,\frac{\mu}{4}}(\{y'\})}$. So, by Property~$(*)^+$ noted above, $y\in R^{\varepsilon}_{f,\frac{\mu}{2}}(\{y'\})$.  But $\nu(y')\geq 0$ and $\nu(y)=0$. Hence   $y'\in\Closure{V}$ and $y\not\in\Interior{V}$.

\end{itemize}
In both cases, we have a $\frac{\mu}{2}$-solution from a point in $\Closure{V}$ to a point not in $\Interior{V}$. Hence, using the same construction as in the second part of the proof of Lemma~\ref{lem:14} above, one gets a $\mu$-solution from a point in $V$ to a point not in $V$ which is a contradiction to the fact that $V$ is a $\mu$-robust safety certificate.







\end{proof}

Using the notation 
\[(\Lie{f}{\nu})(x) := \lim_{h\rightarrow 0} \frac{\nu(\Flow_f(x, h))-\nu(x)}{h}\]
we now remind the reader of a classical result by W.~Wilson~\cite[Theorem~2.5]{Wilson:69}, and for the reader's convenience, we state here a direct corollary: 

\begin{lemma}
\label{lem:15}
  If $h$ is a continuous real-valued function on a compact set $N$, $f$ is a non-singular vector field on $N$, and $\Lie{f}{h}$ is defined on $N$ and continuous, then for every $\varepsilon>0$ there exists a $C^{\infty}$-function $h': N\rightarrow \mathbb{R}$ such that for all $x\in N$, 
\[ |h(x)-h'(x)|<\varepsilon\]
and
\[ (\Lie{f}{h'})(x) > (\Lie{f}{h})(x)-\varepsilon.\]
\end{lemma}

Based on this, we can now construct a robust barrier certificate from the robust safety certificate $V$:

\begin{lemma}
  There is a  differentiable function  $\beta:\mathbb{R}^n\rightarrow\mathbb{R}$ and $\varepsilon>0$ such that
\begin{itemize}
\item $\forall x \;.\; x\in I\Rightarrow \beta(x)>\varepsilon$
\item $\forall x \;.\; |\beta(x)|\leq\varepsilon\Rightarrow L_f(\beta)(x)>\varepsilon$, that is $(\nabla \beta(x))^T f(x) >\varepsilon$
\item $\forall x \;.\; x\in U\Rightarrow \beta(x)<\varepsilon$
\end{itemize}

\end{lemma}

\begin{proof}
From the $\varepsilon$-robust safety certificate $V$ we can construct a function $\nu$ as above. The zero set $\{ x \mid \nu(x)=0\}$ is equal to $\Frontier{V}$, and hence 
it intersects neither $I$ nor $U$. Moreover, $f$ is non-singular on this set. For all $\delta\in\mathbb{R}^{\geq 0}$, let $N_\delta:=\{ x \mid |\nu(x)| \leq \delta\}$. Now let $0<\delta<\frac{1}{2}$ be such that the $\delta$-sublevel set $N_\delta$ still intersects neither $I$ nor $U$, $f$ is non-singular on $N_\delta$, and $N_\delta$ is contained in the open neighborhood of $\Frontier{V}$ for which Lemma~\ref{lem:16} ensures continuity of $\nu$.

Now observe that for all $x\in N_\delta$,
\begin{align*}
 (\Lie{f}{\nu})(x)&=&\\
\lim_{h\rightarrow 0} \frac{\nu(\Flow_f(x, h))-\nu(x)}{h}&=&\\
\lim_{h\rightarrow 0} \frac{\nu(x)+h-\nu(x) }{h}&=&1.
\end{align*}

Further, due to continuity of $\nu$ we can use Lemma~\ref{lem:15} to construct a 
$C^{\infty}$-function $\nu'$ such that for all $x\in N_\delta$, \[ |\nu(x)-\nu'(x)|<\frac{\delta}{2}\]
and
\[ (\Lie{f}{\nu'})(x) > (\Lie{f}{\nu})(x)-\frac{\delta}{2},\]

Then for all $x\in N_\delta$, $(\Lie{f}{\nu'})(x)>1-\frac{\delta}{2}>0$, and for all $x\not\in N_{\frac{\delta}{2}}$, $\nu'(x)\neq 0$.

Now let $\mathcal{X}=(X_\alpha)_{\alpha\in A}$ be an open cover of $\mathbb{R}^n$ indexed by a set $A$. Assume furthermore, that for all $\alpha\in A$, $X_\alpha$ contains a common element with at most one of the sets $N_{\frac{\delta}{2}}$, $\{ x \mid \nu(x)\geq \delta\}$, and $\{ x \mid \nu(x)\leq -\delta\}$. Let $(\beta_\alpha)_{\alpha\in A}$ be a family of functions s.t. for each $\alpha\in A$, and $x\in\mathbb{R}^n$, 
\[\beta_{\alpha}(x):=\left\{ \begin{array}{ll}
                                         {\delta}, & \mbox{if } X_{\alpha}\cap \{ x \mid \nu(x)\geq \delta\}\neq\emptyset,\\[0.2ex]
                               {-\delta}, & \mbox{if } X_{\alpha}\cap \{ x \mid \nu(x)\leq -\delta\}\neq\emptyset,\\[0.2ex]
                                         {\nu'(x)}, & \mbox{otherwise}.\\[0.2ex]                               
                                          \end{array} \right.   \]

                                      Let $(\psi_{\alpha})_{\alpha\in A}$ be a partition of unity~\cite[Chapter 2]{Lee:12} sub-ordinate to $\mathcal{X}$ and \[\beta(x):= \sum_{\alpha\in A} \psi_{\alpha}(x) \beta_{\alpha}(x). \] 

Since $(\psi_{\alpha})_{\alpha\in A}$ is a partition of unity, $\beta$ is smooth. Since for all $\alpha\in A$, for all $x\not\in N_{\frac{\delta}{2}}$, $\beta_\alpha(x)\neq 0$, $\beta$ is                                  
non-zero outside of $N_{\frac{\delta}{2}}$. Since $\beta$ coincides with $\nu'$ on $N_{\frac{\delta}{2}}$, for all $x\in N_{\frac{\delta}{2}}$, $(\Lie{f}{\beta})(x)>0$. Moreover, $\beta$ is strictly positive on $I$, and strictly negative on $U$. The existence of a bound $\varepsilon>0$ follows due to continuity of $\beta$ and compactness of all relevant sets.


\end{proof}

Theorem~\ref{thm:3} follows as a corollary.

\section{Finite Time Converse Theorem for Safety Certificates}
\label{sec:finite-time}

The construction in the previous section was based on a 
robust safety certificate $V$ whose existence is ensured by Property~\ref{prop:converse_robust}. However, this property has the disadvantage of being based on an infinite time reach set. We will remove this disadvantage by proving Theorem~\ref{thm:1} in the next section.

The proof extends a technique introduced by M.~Fr{\"a}nzle~\cite{Fraenzle:99}, that he originally used for proving existence of a barrier for hybrid systems with polynomial flow: Take the bloated finite-time reach set $R_{f,\varepsilon}^{[0, t]}(I)$, and show that the original dynamics is shrinking on it if the bloated  reach set does not grow beyond the bounded complement of $U$. However, here we have additional complications: First, our flow is, in general, not polynomial (note that even linear ODEs usually have a non-polynomial flow), and second, we are not satisfied with a barrier, but we want a \emph{robust} barrier.


Before the actual proof, we note:

\begin{lemma}
  \label{lem:2}
  Let $V\subseteq\mathbb{R}^n$, $\Delta>0$ such that \[R_{f,\varepsilon}^{\Delta}(V)\subseteq V.\] Then for all $\Delta'\geq\Delta$, \[R_{f,\varepsilon}(R_{f,\varepsilon}^{[0, \Delta']}(V))\subseteq R_{f,\varepsilon}^{[0, \Delta']}(V). \]
\end{lemma}

\begin{proof}
  Let $p\in R_{f,\varepsilon}(R_{f,\varepsilon}^{[0, \Delta']}(V))$. We prove that $p$ is also in $R_{f,\varepsilon}^{[0, \Delta']}(V)$.

  Let $\sigma$ be the $\varepsilon$-solution leading from a point in $V$ to $p$, and let $\tau$ be the length of this solution. So $\sigma(0)\in V$ and $\sigma(\tau)=p$.  Since $R_{f,\varepsilon}^{\Delta}(V)\subseteq V$, also $\sigma(\Delta)\in V$. In general, for every $i$, $\sigma(i\Delta)\in V$.  Let $b$ be the maximal $i$ with $i\Delta<\tau$. Then $0\leq \tau-b\Delta\leq\Delta\leq\Delta'$, and hence
$p\in R_{f,\varepsilon}^{\tau-b\Delta}(\{ \sigma(b\Delta) \})\subseteq R_{f,\varepsilon}^{[0, \Delta']}(\{ \sigma(b\Delta) \})\subseteq R_{f,\varepsilon}^{[0, \Delta']}(V)$. 
\end{proof}



Now we are ready to prove Theorem~\ref{thm:1}:

\begin{proof}
We assume that $(f, I, U)$ is robustly safe with robustness margin $\varepsilon$, and the complement of $U$ is bounded. We derive a contradiction from the additional assumption that  
 there exists $\Delta>0$ such that for all  $t\geq 0$, $R_{f,\frac{\varepsilon}{2}}^{[0,\Delta]}(R_{f,\varepsilon}^{[0,t]}(I))$ is no $\frac{\varepsilon}{2}$-robust safety certificate. 

 We first observe that for every $t\geq 0$, $R_{f,\varepsilon}^{[0,t]}(I)\not\supseteq R^\Delta_{f, \frac{\varepsilon}{2}}(R_{f,\varepsilon}^{[0,t]}(I))$. Otherwise,  due to Lemma~\ref{lem:2},
$R_{f,\frac{\varepsilon}{2}}(R_{f,\frac{\varepsilon}{2}}^{[0, \Delta]}(  R_{f,\varepsilon}^{[0,t]}(I)))\subseteq R_{f,\frac{\varepsilon}{2}}^{[0, \Delta]}(R_{f,\varepsilon}^{[0,t]}(I))$, and 
 $R_{f,\frac{\varepsilon}{2}}^{[0,\Delta]}(R_{f,\varepsilon}^{[0,t]}(I))$ would be a $\frac{\varepsilon}{2}$-robust safety certificate, which would be in contradiction to the assumption above.


We now derive a contradiction from the above by constructing an infinite sequence of points $p_1,\dots,$ such that 
\begin{itemize}
\item for every $i$, $p_i\in R_{f,\varepsilon}^{[0, \infty]}(I)$, and
\item for every $i,j$, $i\neq j$, we have $||p_j-p_i||>\delta$, 
\end{itemize}
with $\delta$ being a positive real number. This implies that the maximal distance of two points in the sequence grows over all bounds. Since the complement of $U$ is bounded, this is a contradiction to the fact that $(f, I, U)$ is robustly safe.

We ensure the second property (sufficient distance of points) by constructing at the same time a sequence $t_1,\dots,$ such that, for all $i$, 
\begin{itemize}
\item every point not $\varepsilon$-reachable in time $t_{i}$, that is, every point not in $R_{f,\varepsilon}^{[0,t_{i}]}(I)$ has distance at least $\delta$ from every point $p_1,\dots,p_{i}$ (in Figure~\ref{fig:proof}, the $\delta$-balls around $p_i$ are inside of the reach sets $R_{f,\varepsilon}^{[0,t_{i}]}(I)$).
\item $p_{i+1}$ is not $\varepsilon$-reachable in time $t_{i}$, that is, $p_{i+1}\not\in R_{f,\varepsilon}^{[0,t_{i}]}(I)$ (in Figure~\ref{fig:proof}, the points $p'=p_{i+1}$ are outside of the reach sets $R_{f,\varepsilon}^{[0,t_{i}]}(I)$).
\end{itemize}

\begin{figure}
  \centering
  \Preprint{\includegraphics[width=8cm,trim={1cm 3cm 5cm 0cm},clip]{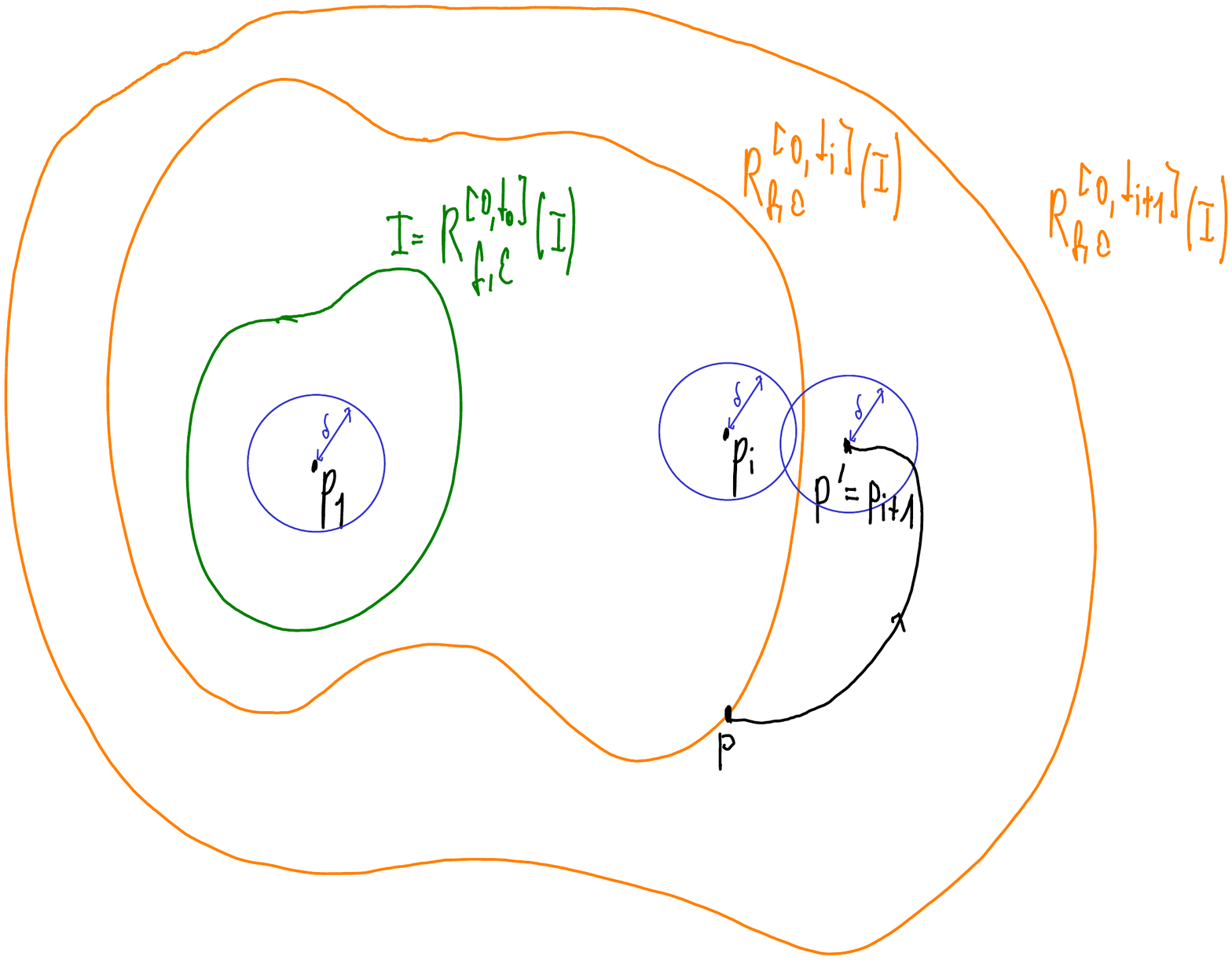}}{  \includegraphics[width=8.2cm,clip]{proof_new.pdf}}
\caption{Proof Intuition}
\label{fig:proof}
\end{figure}

It remains to construct the sequences $p_1,\dots,$ and $t_1,\dots,$ in such a way. Due to the Heine-Borel theorem, the closure of the complement of $U$ is compact. So we can apply Lemma~\ref{lem:1}: Let $\delta$ be as given by Lemma~\ref{lem:1} for $\Omega$ being the closure of the complement of $U$. Let $i$ be an arbitrary, but fixed natural number, and assume that the above properties hold for $p_1,\dots,p_i$ and $t_1,\dots, t_i$. We construct $p_{i+1}, t_{i+1}$ together with a proof that the above properties hold for them. 

For this, let $p\in R_{f,\varepsilon}^{[0,t_i]}(I)$, $p'\not\in R_{f,\varepsilon}^{[0,t_i]}(I)$ s.t. $p'$ is $\frac{\varepsilon}{2}$-reachable in time $\Delta$ from $p$ which exist since for all $t\geq 0$, $R_{f,\varepsilon}^{[0,t]}(I)\not\supseteq R^\Delta_{f, \frac{\varepsilon}{2}}(R_{f,\varepsilon}^{[0,t]}(I))$. By  Lemma~\ref{lem:1}  every point in the $\delta$-neighborhood of $p'$ is reachable in time $\Delta$ from $p\in R_{f,\varepsilon}^{[0,t_i]}(I)$ with $\varepsilon$-perturbed dynamics. 

So we choose $p_{i+1}=p', t_{i+1}= t_i+\Delta$. Then $p_{i+1}\not\in R_{f,\varepsilon}^{[0,t_{i}]}(I)$, and hence has distance at least $\delta$ from $p_1,\dots, p_i$. Moreover, every point not reachable in time $t_{i+1}$ again has distance at least $\delta$ from $p_1,\dots, p_{i+1}$.

We choose $t_1=0$ and $p_1$ an initial point that has distance at least $\delta$ from the boundary of the set of initial points $I$ (the case where such a point does not exist can be easily handled by shifting the sequence).






\end{proof}



\section{Conclusion}
\label{sec:conclusion}

In this paper we proved a robust converse theorem for barrier certificates, removing restrictions of such theorems in the literature. Moreover, we proved that a certain finite time reach set of a robustly safe system always forms a robust safety certificate. The literature contains several variants of the notion of a barrier certificate~\cite{Taly:09,Ghorbal:16}. We propose a comprehensive converse theory covering those variants as a goal for future work in the area. 

The actual computation of barrier certificates~\cite{Dai:17,Djaballah:17,Yang:15,Ratschan:17} is an essential question that deserves further attention.

\bibliographystyle{plain}
\bibliography{sratscha}


\end{document}